\documentclass[fleqn
,9pt
]{sig-alternate}

\usepackage{etex}

\usepackage[utf8]{inputenc} % Codage du fichier
\usepackage{mathptmx}
\usepackage[T1]{fontenc}

\usepackage[english]{babel}

\makeatletter
\def\@citex#1[#2]#3{% 
  \@safe@activestrue\edef\@tempd{#1}\@safe@activesfalse
  \@safe@activestrue\edef\@tempe{#3}\@safe@activesfalse
  \org@@citex{\@tempd}[#2]{\@tempe}}%
\makeatother

\usepackage[final]{microtype}

\usepackage[format=hang]{caption}
\usepackage[font=small]{subfig}

\usepackage{mathtools}
\mathtoolsset{%
  showonlyrefs,% On ne veut afficher que les étiquettes des équations auxquelles on se réfère ensuite
  showmanualtags}% et celles dont on spécifie manuellement l'étiquette (par exemple des axiomes)

\usepackage{array}

\usepackage{verbatim}

\usepackage{eqparbox}

\usepackage[binary-units=true]{siunitx}

\usepackage{filecontents}

\usepackage[ruled,lined,noend,linesnumbered]{algorithm2e}

\usepackage{multirow}

\usepackage{xfrac}

\usepackage{tikz}
\usetikzlibrary{arrows,calc,3d,matrix,%intersections,
plotmarks}
\usepackage{pgfplots}
\pgfplotsset{compat=1.5}
\usepackage{enumitem} % Pour les listes dans les théorèmes

\usepackage[colorlinks,linkcolor=blue%
,draft                           % Comment this line and uncomment the next one for an on-screen version  
,plainpages=true,pdfpagelabels%  % Pour résoudre les avertissements "destination with the same identifier (name{page.}) has been already used..."
]{hyperref}
\hypersetup{
  pdfinfo={
    Title={On the Complexity of Computing Gröbner Bases For Quasi-Homogeneous Systems}
  }
}

\usepackage{xspace}

\newtheorem{theorem}{Theorem}
\newtheorem{lemme}[theorem]{Lemma}
\newtheorem{prop}[theorem]{Proposition}

\newdef{definition}%[theorem]
{Definition}
\newdef{remark}{Remark}

\newcommand{\N}{\mathbb{N}}
 
\newcommand{\K}{\mathbb{K}}
\newcommand{\FF}{\mathbb{F}} % \F est déjà utilisé pour F5

\renewcommand{\tilde}[1]{\widetilde{#1}}
\renewcommand{\epsilon}{\varepsilon}
\renewcommand{\phi}{\varphi}

\SetArgSty{textrm} % Style pour le texte dans les algorithmes
\SetKw{KwDownTo}{downto}

\newcommand{\grevlex}{\textsc{GRevLex}\xspace}
\newcommand{\lex}{\textsc{Lex}\xspace}
\newcommand{\Wgrevlex}{$W$-\grevlex}

\newcommand{\ltgrevlex}{<_\textrm{grevlex}}
\newcommand{\ltWgrevlex}{<_\textrm{$W$-grevlex}}

\newcommand{\LT}{\mathrm{LT}}
\newcommand{\Spol}{\mathrm{S\text{-}Pol}}

\newcommand{\F}[1]{\ifmmode\mathrm{F_{#1}}\else\textrm{F}$_{#1}$\fi\xspace}

\newcommand{\HS}{\mathrm{HS}}
\newcommand{\HI}{i_{\mathrm{reg}}}
\newcommand{\ireg}{\HI}
\newcommand{\dreg}{d_{\mathrm{reg}}}

\newcommand{\Fext}{F_{\mathrm{ext}}}

\newcommand{\pgcdfam}[1]{\mathrm{gcd}\left(#1\right)}

\newcommand{\ordre}[1]{<_{#1}}

\newcommand{\append}{\,\cup} % Ajout d'une ligne à une matrice

\newcommand{\ifnonempty}[3]{% Teste si le premier argument est vide, si oui renvoie le 3e, sinon le 2e
\def\tempa{}%
\def\tempb{#1}%
\ifx\tempa\tempb % Cas vide
  #3 
\else % Cas non vide
  #2
\fi}

\renewcommand{\hom}[2][]{%
  \hspace{0.1em}%
  \mathrm{hom}_{#2}^{\ifnonempty{#1}{#1}{\phantom{-1}}}
  \ifnonempty{#1}{\hspace{-0.1em}}{\hspace{-0.2em}}%
}

\newcommand{\D}{\mathbf{D}}

\usepackage{ifdraft}

\newcommand{\myvspace}[1]{\vspace{#1}}

\newenvironment{my_enumerate}{%
\vspace{-6pt}
\begin{enumerate}
     \setlength{\itemsep}{0pt}
     \setlength{\parskip}{0pt}
     \setlength{\parsep}{0pt}}
{\end{enumerate}
\vspace{-6pt}
}

\makeatletter
\def\hypCounter#1{\expandafter\@hypCounter\csname c@#1\endcsname}
\def\@hypCounter#1{\ifcase#1\or H1\or H2\or H2'\fi}
\AddEnumerateCounter{\hypCounter}{\@hypCounter}{3}
\makeatother

\newenvironment{my_newenumerate}{%
\vspace{-6pt}
\begin{enumerate}[align=left,label=\textbf{\hypCounter*.},ref=\textbf{\hypCounter*},resume]
     \setlength{\itemsep}{0pt}
     \setlength{\parskip}{0pt}
     \setlength{\parsep}{0pt}}
{\end{enumerate}
\vspace{-6pt}
}

\newenvironment{my_itemize}{%
\vspace{-6pt}
\begin{itemize}
     \setlength{\itemsep}{0pt}
     \setlength{\parskip}{0pt}
     \setlength{\parsep}{0pt}}
{\end{itemize}
\vspace{-6pt}
}

\newcommand{\superscript}[1]{\ensuremath{^{\textrm{#1}}}}
\def\sharedaffiliation{\end{tabular}\newline\begin{tabular}{c}}

\def\labo{\superscript{*}}
\def\ens{\superscript{\dag}}
\def\iuf{\superscript{\ddag}}

\title{On the Complexity of Computing Gröbner Bases\\for Quasi-Homogeneous Systems}

\numberofauthors{3} % Inutile

\author{%
  \and \ Jean-Charles Faugère\labo\\
  \email{Jean-Charles.Faugere@inria.fr}
  \and \ Mohab Safey El Din\labo\iuf\\
  \email{Mohab.Safey@lip6.fr}
  \and \ Thibaut Verron\ens\labo\\
  \email{Thibaut.Verron@ens.fr}
  \sharedaffiliation
  \begin{tabular}{ccc}
    \affaddr{{\labo}INRIA, Paris-Rocquencourt Center, PolSys Project{\ }} && \affaddr{{\iuf}Institut Universitaire de France{\ }}   \\
    \affaddr{UPMC, Univ. Paris 06, LIP6} &&  \affaddr{% Cedex 05, France
    } \\
    \affaddr{CNRS, UMR 7606, LIP6} && \affaddr{{\ens}École Normale Supérieure,} \\
    \affaddr{Case 169, 4, Place Jussieu, F-75252 Paris} &&  \affaddr{45, rue d'Ulm, F-75230, Paris} \\
  \end{tabular}
}

\begin{document}
\conferenceinfo{ISSAC'13,} {June 26--29, 2013, Boston, Massachusetts, USA.}
\CopyrightYear{2013}
\crdata{978-1-4503-2059-7/13/06}
\clubpenalty=10000
\widowpenalty = 10000

\addtolength\abovedisplayskip{-0.5\baselineskip}%
\addtolength\belowdisplayskip{-0.5\baselineskip}%

\maketitle

\begin{abstract}
  Let $\K$ be a field and $(f_1, \ldots, f_n)\subset \K[X_1, \ldots, X_n]$ be a sequence of quasi-homogeneous polynomials of respective weighted degrees $(d_1, \ldots, d_n)$ w.r.t a system of weights $(w_{1},\dots,w_{n})$.
  Such systems are likely to arise from a lot of applications, including physics or cryptography.
  
  We design strategies for computing Gr\"obner bases for quasi-homogeneous systems by adapting existing algorithms for homogeneous systems to the quasi-homogeneous case.
  Overall, under genericity assumptions, we show that for a generic zero-dimensional quasi-homogeneous system, the complexity of the full strategy is polynomial in the weighted Bézout bound $\prod_{i=1}^{n}d_{i} / \prod_{i=1}^{n}w_{i}$.

  We provide some experimental results based on generic systems as well as systems arising from a cryptography problem.
  %They show that taking advantage of the quasi-homogeneous structure of the systems can speed the computations up, by a factor of up to 10 for the real-life system.
  They show that taking advantage of the quasi-homogeneous structure of the systems allow us to solve systems that were out of reach otherwise.
\end{abstract}

%\terms{Algorithms}

\section{Introduction}
\label{sec:introduction}
Polynomial system solving is a very important problem in computer algebra, with a wide range of applications in theory (algorithmic geometry) or in real life (cryptography).
For that purpose, Gröbner bases of polynomial ideals are a valuable tool, and practicable computation of the Gröbner bases of any given ideal is a major challenge of modern computer algebra.
Since their introduction in 1965, many algorithms have been designed to compute Gröbner bases (\cite{Buch76,F99a,Fau02a,FGLM}), improving the efficiency of the computations.

Systems arising from ``real life'' problems often have some structure.
It has been observed that most of these structures can make the Gröbner basis easier to compute.
For example, it is known that homogeneous systems, or systems with an important maximal homogeneous component, are better solved by using a degree-compatible order, and then applying a change of ordering.
In this paper, we study a structure slightly more general than homogeneity, called \emph{quasi-homogeneity}.
More precisely, we will say that a polynomial $P(X_{1},\dots,X_{n})$ is quasi-homogeneous for the \emph{system of weights} $W=(w_{1},\dots,w_{n})$, if the polynomial
\begin{equation}
  Q(Y_{1},\dots,Y_{n}):=P(Y_{1}^{w_{1}},\dots,Y_{n}^{w_{n}})\label{eq:7}
\end{equation}
is homogeneous.
Systems with such a structure are likely to arise for example from physics, where all measures are associated with a dimension which, to some extent, can be seen as a weight.

Let $F = (f_{1},\dots,f_{m})$ be a system of polynomials, in a polynomial algebra graded w.r.t the system of weights  $W=(w_{1},\dots,w_{n})$.
In the following, we will assume that $F$ is quasi-homogeneous and generic, or more generally that its quasi-homogeneous components of maximal weighted degree are generic.
It is possible to compute directly a Gröbner basis of the ideal generated by $F$ .
This strategy consists of running the classical algorithms \F5~(\cite{Fau02a}) and FGLM~(\cite{FGLM}) on $F$, while ignoring the quasi-homogeneous structure.
However, to the best of our knowledge, there is no general way of evaluating the complexity of that strategy.

Another approach is to compute the \emph{homogenized} system defined by
%\begin{equation}
$\tilde{F} := (f_{i}(X_{1}^{w_{1}},\dots,X_{n}^{w_{n}}))$, %espace
%\end{equation}
and then compute a Gröbner basis of that system, using the usual strategies for the homogeneous structure.
Experimentally, the first step of the computation is much faster than with the naive strategy.
However, the number of solutions is increased by a factor of $\prod_{i=1}^{n}w_{i}$, slowing down the change of ordering, which thus becomes the main bottleneck of the computation.

Furthermore, to the best of our knowledge, the best complexity bounds for this computation are those we obtain for a homogeneous system of the same degree.
However, experimentally, the first step of the computation proves faster for a homogenized quasi-homogeneous system with weighted degree $(d_{1},\dots,d_{n})$ than for a homogeneous system of total degree $(d_{1},\dots,d_{n})$.

\vspace{-10pt}
\paragraph*{Main results}We provide a complexity study of the above strategy, allowing us to quantify %?
this speed-up, as well as to propose a workaround for the change of ordering.
Overall, we prove that the known bounds for this strategy can be divided by $\prod_{i=1}^{n}w_{i}$ for a generic zero-dimensional $W$-homogeneous system with weights $W=(w_{1},\dots,w_{n})$.

More precisely, we assume the system $(f_{1},\dots,f_{m})$ to satisfy the two following generic assumptions:
\begin{my_newenumerate}
  \item \label{item:H1}The sequence $f_{1}, \dots, f_{m}$ is regular;
  \item \label{item:H2}The sequence $f_{1},\dots,f_{i}$ is in Noether position w.r.t. $X_{1},\dots,X_{i}$, for any $1 \leq i \leq m$.
\end{my_newenumerate}

Under hypothesis~\ref{item:H1}, we adapt the classical results of the homogeneous case, using similar arguments based on Hilbert series, to estimate the degree of the ideal and the degree of regularity of the system:
\begin{equation}
  \label{eq:11}
  \deg(I) 
  = \prod_{i=1}^{n}\frac{d_{i}}{w_{i}}\; ;
  \;\;\;
  \dreg(F) \leq \sum_{i=1}^{n}\Big(d_{i} - w_{i}\Big) + \max\{w_{i}\}.
\end{equation}

We study the complexity of the \F5 algorithm through its matrix variant matrix-\F5.
This is a usual approach, carried on for example in~\cite{FR09}.
With minor changes, the matrix-\F5 algorithm for homogeneous systems can be adapted to quasi-homogeneous systems.
A combinatorial result found in~\cite{Geir02} shows that the number of columns of the matrices appearing in that variant of matrix-\F5 is approximately smaller by a factor of $\prod_{i=1}^{n}w_{i}$, when compared to the regular matrix-\F5 algorithm.
Overall, we can obtain complexity bounds which are smaller by a factor of $P^{\omega}$ than the bounds we would obtain for a generic homogeneous system with same degrees, where $P = \prod_{i=1}^{n}w_{i}$ and $\omega$ is the exponent of the complexity of matrix multiplication.
In the end, we show that for systems satisfying~\ref{item:H1}, our strategy, running \F5 on the homogenized system, dehomogenizing the result, and then running FGLM, performs in time polynomial in $\prod_{i=1}^{n}d_{i}/\prod_{i=1}^{n}w_{i}$, that is polynomial in the number of solutions.

Further assuming hypothesis~\ref{item:H2}, we also carry on the precise complexity analyses done in~\cite{Bar04} for homogeneous systems, and adapt them to the quasi-homogeneous case to deduce a precise complexity bound for our quasi-homogeneous variant of Matrix-\F5.
These new complexity bounds are also smaller by a factor of $P^{\omega}$ than similar bounds for a generic homogeneous system.
Even though these bounds still do not match exactly the experimental complexity, they tend to confirm that overall, we are able to compute a \lex Gröbner basis for a generic quasi-homogeneous system in time reduced by a factor of $P^{\omega}$, when compared with a generic homogeneous system with same degrees.

We have run benchmarks with the FGb library (\cite{F10c}) and the Magma computer algebra software (\cite{Magma}), on both generic systems and real-life systems arising in cryptography.
Experimentally, in both cases, our strategy seems always faster than ignoring the quasi-homogeneous structure, and the speed-up increases with the considered weights. 

Experiments have also shown that the order of the variables can have an impact on the performances of both strategies.
Predicting this behavior seems to require more sophisticated tools and may be material for future research.

\vspace{-10pt}
\paragraph*{Prior works}
\label{sec:prior-works}
Making use of the structure of polynomial systems to develop faster algorithms has been a general trend over the past few years:
see for example~\cite{FJ03}, \cite{MR2035232} or \cite{FSS10b}. 
Polynomial algebras graded with respect to a system of weights have been studied by researchers in commutative algebra.
Most notably, the Hilbert series of ideals defined by regular sequences, which we use several times in this paper, is well known, and could be found for example in~\cite{Stanley78}.
The paper~\cite{robbiano86} defines many structures of polynomial algebras, including weighted gradings, in preparation for future algorithmic developments.
Combinatorial objects arising when we try to estimate the number of monomials of a given $W$-degree are called \emph{Sylvester denumerants}, and studied for example in~\cite{Geir02}.

When it comes to Gröbner bases, weighted gradings and related orderings have been described in early works such as~\cite{becker1993grobner}.
However, as far as we know, the impact of a quasi-homogeneous structure on the complexity of Gröbner bases computations had never been studied.% Dire ici que ça n'étonne personne?

Among the various computer algebra software able to compute Gröbner bases, it seems that only Magma has algorithms dedicated to quasi-homogeneous systems.
Given a quasi-homogeneous system, it will detect the appropriate system of weights, and use the \Wgrevlex ordering to compute an intermediate basis before the change of ordering.
However, this strategy is only available for quasi-homogeneous systems, while it can be useful in many other cases, for example systems of polynomials defined as the sum of a quasi-homogeneous component and a scalar.

Other computer algebra software (e.g. Singular) allow the user to compute $\tilde{F}$ and to run the Gröbner basis algorithm on it.
Since all these algorithms (most often Buchberger, \F4 or \F5) use $S$-pairs, they will show a similar speed-up.
However, the user must notice that the computations may benefit from using a quasi-homogeneous structure of the system, and provide the system of weights.

We do not provide a way to know what is the ``appropriate'' system of weights for a given system, or even to detect systems which would benefit from taking into account the quasi-homogeneous structure.
However, some systems obviously belong to that category (e.g quasi-homogeneous plus scalar), and the system of weights will then be easy to compute.

\vskip-0.35cm
\paragraph*{Structure of the paper}
\label{sec:structure-paper}
In section~\ref{sec:nouvelle-section-2} we define more precisely quasi-homogeneous systems, and we compute their degree and degree of regularity assuming the above hypotheses.
We also take this opportunity to show briefly that these hypotheses are generic.
In section~\ref{sec:comp-grobn-basis} we prove that the strategy consisting of modifying the system is correct, we explain how we
can adapt matrix-\F5 and FGLM to quasi-homogeneous systems, and then we evaluate the complexity of these algorithms.
In section~\ref{sec:cas-affine}, we briefly explain how these results for quasi-homogeneous systems still help in case the system was obtained from a quasi-homogeneous system by specializing one of the variables to 1.
We also give an example of such a structure, as well as the associated algorithm.
Finally, in section~\ref{sec:experimental-results}, we give some experimental results.

\vspace{-0.2cm}
\section{Quasi-homogeneous systems}
\label{sec:nouvelle-section-2}

\vspace{-0.2cm}
\subsection{Weighted degrees and polynomials}
\label{sec:polynomes-degres}

\noindent
Let $\K$ be a field.
We consider the algebra $A := \K[X_{1},\dots,X_{n}] = \K[\mathbf{X}]$.
Even though one usually uses the total degree to grade the algebra $A$, there are other ways to define such a grading, as seen in \cite{becker1993grobner}, for example.

\vspace{-0.2cm}
\begin{definition}
  Let $W = (w_{1},\dots,w_{n})$ be a vector of positive integers.
  Let  $\alpha = (\alpha_{1}, \dots, \alpha_{n})$ be a tuple of nonnegative integers.
  Let the integer
  $ \deg_{W}(\mathbf{X^{\alpha}}) = \sum_{i=1}^{n} w_{i}\alpha_{i}$
  be the \emph{$W$-degree}, or \emph{weighted degree} of the monomial
  $\mathbf{X}^{\alpha} = X_{1}^{\alpha_{1}}\cdots X_{n}^{\alpha_{n}}.$
  Call the vector $W$ a \emph{system of weights}.
  We denote by $\mathbf{1}$ the system of weights defined by $(1,\dots,1)$, associated with the usual grading on $A$.
\end{definition}

\vspace{-0.2cm}
One can prove that any grading on $\K[\mathbf{X}]$
comes from such a system of weights (\cite[sec.~10.2]{becker1993grobner}).
We denote by $(\K[\mathbf{X}],W)$ the $W$-graded structure on $A$, and
in that case, to clear ambiguities, we use the adjective
\emph{$W$-homogeneous} for elements or ideals, or
\emph{quasi-homogeneous} or \emph{weighted homogeneous} if $W$ is clear in the context.  The word
\emph{homogeneous} will be reserved for $\mathbf{1}$-homogeneous items.

\vspace{-0.2cm}
\begin{prop}
  Let $(\K[X_{1},\dots,X_{n}],W)$ be a graded polynomial algebra.
  Then the application
    \begin{equation}
    \label{eq:3}
    \begin{matrix}
      \hom{W} : & (\K[X_{1},\dots,X_{n}],W) & \to & (\K[t_{1},\dots,t_{n}],\mathbf{1}) \\
      & f & \mapsto & f(t_{1}^{w_{1}},\dots,t_{n}^{w_{n}})
    \end{matrix}
  \end{equation}
  is an injective graded morphism, and in particular the image of a quasi-homogeneous polynomial is a homogeneous polynomial.
\end{prop}
\begin{proof}
  It is an easy consequence of the definition of the grading w.r.t a system of weights.
\end{proof}\vspace{-0.2cm}

The above morphism also provides a quasi-homogeneous variant of the \grevlex ordering (as found for example in \cite{becker1993grobner}), which we call the \emph{\Wgrevlex ordering}:
\begin{equation}
  \label{eq:102}
      u \ltWgrevlex v
      \iff \hom{W}(v) \ltgrevlex \hom{W}(v)
\end{equation}
Given a $W$-homogeneous system $F$, one can build the homogeneous system $\hom{W}(F)$, and then apply classical algorithms (\cite{Fau02a,FGLM}) to that system to compute a \grevlex (resp. \lex{}) Gröbner basis of the ideal generated by $\hom{W}(F)$.
We will prove in section~\ref{sec:comp-grobn-basis} (prop.~\ref{lemme:lemme:passage_par_jW})
that this basis is contained in the image of $\hom{W}$, and that its pullback is a \Wgrevlex (resp. \lex{}) Gröbner basis of the ideal generated by $F$.

Let us end this paragraph with some notations and definitions.
The \emph{degree of regularity} of the system $F$ is the highest degree $\dreg(F)$ reached in a run of \F5 to compute a \grevlex Gröbner basis of $\hom{W}(F)$.
The \emph{index of regularity} of an ideal $I$ is the degree $\ireg$ of the Hilbert series $\HS_{A/I}$, defined as the difference of the degree of its numerator and the degree of its denominator.

Recall that given a homogeneous ideal $I$, we define its degree $D$ as the degree of the projective variety $V(I)$, as introduced for example in~\cite{hartshorne77}.
This definition still holds for the quasi-homogeneous case.
In case the projective variety is empty, that is if the affine variety is equal to $\{0\}$, we extend that definition by letting $D$ be the multiplicity of the $0$ point, that is the dimension of the $\K$-vector space $A/I$.
Finally, from now on we will only consider \emph{affine} varieties, even when the ideal is quasi-homogeneous.
In particular, the dimension of $V(0)$ is $n$, and that a zero-dimensional variety will be defined by at least $n$ polynomials.

\subsection{Degree and degree of regularity}
\label{sec:degr-regul-degr}

\vspace{-0.4cm}
\paragraph*{Zero-dimensional regular sequences}
\label{sec:hilbert-series-w}
As in the homogeneous case, regular sequences are an important case to study, because it is a generic property which allows us to compute several key parameters and good complexity bounds.
We first characterize the degree and bound the degree of regularity of a zero-dimensional ideal defined by a regular sequence.

\vspace{-0.3cm}
\begin{theorem}\label{thm:zerodim-values}
  Let $W = (w_{1},\dots,w_{n})$ be a system of weights, and $F=(f_{1},\dots,f_{m})$ a regular sequence of $W$-homogeneous polynomials, of respective $W$-degrees $d_{1},\dots,d_{m}$.
  Further assume that the set of solutions is zero-dimensional, that is $m=n$.
  We denote by $I$ the quasi-homogeneous ideal generated by $F$. 
  Then we have $\deg(I) =\prod_{i=1}^{n}\frac{d_{i}}{w_{i}}$ and $\dreg(F) \leq \sum_{i=1}^{n}\big(d_{i} - w_{i}\big) + \max\{w_{i}\}$.
\end{theorem}
\begin{proof}
  We will determine the degree and degree of regularity of the system from the Hilbert series (or Poincaré series) of the algebra $A/I$.
  A classical result which can be found for example in~\cite[cor.~3.3]{Stanley78} states that, for regular sequences, this series is
  \begin{equation}
    \label{eq:2}
    \HS_{\sfrac{A}{I}}(t)=\frac{(1-t^{d_{1}})\cdots(1-t^{d_{m}})}{(1-t^{w_{1}})\cdots(1-t^{w_{n}})}.
  \end{equation}
  We assumed $n=m$, so the Hilbert series can be rewritten as
  \begin{equation}
    \label{eq:8}
    \HS_{\sfrac{A}{I}}(t) = \frac{(1+\dots+t^{d_{1}-1})\cdots(1+\dots+t^{d_{n}-1})}{(1+\dots+t^{w_{1}-1})\dots(1+\dots+t^{w_{n}-1})}.
  \end{equation}
  In the 0-dimensional case, recall that the Hilbert series is
  actually a polynomial, and has degree
  $\ireg =\sum_{i=1}^{n}(d_{i}-w_{i})$.
  This means that all monomials of $W$-degree greater than $\ireg$ are in the ideal, and 
  as such, that the leading terms of the \Wgrevlex Gröbner basis of $F$ need to divide all the monomials of $W$-degree greater than $\ireg$.%$d_{\max}$.
  This proves that all the polynomials in the Gröbner basis computed by \F5 have $W$-degree at most
  $ %d_{\max}
  \ireg +\max\{w_{i}\}$.
  And since the \F5~criterion (\cite{Fau02a}) ensures that there is no reduction to zero in a run of \F5 on a regular sequence, the algorithm indeed stops in degree at most $\ireg + \max\{w_{i}\}$.

  Furthermore, the degree of the ideal $I$ is equal to the dimension of the vector space $A/I$%(see \cite{FGLM})
  , that is the value of
  the Hilbert series at $t=1$, that is
  $\prod_{i=1}^{n} \frac{d_{i}}{w_{i}}$.
\end{proof}%\vspace{-0.2cm}

Note that except for this inequality, not much is known about the
degree of regularity of a quasi-homogeneous system.  In particular, the
above bound is nothing more than a bound, even in the generic case.
Let us introduce some examples of the three cases one can observe with
a quasi-homogeneous generic system:
\begin{my_enumerate}
  \item \label{item:1} $W = (3,2,1)$, generic system of $W$-degree $\D=(6,6,6)$: then $\dreg = \ireg +1 = 13$;
  \item \label{item:2} $W = (1,2,3)$, generic system of $W$-degree $\D=(6,6,6)$: then $\dreg = 15 > \ireg +1 = 13$;
  \item \label{item:3} $W=(2,3)$, generic system of $W$-degree $\D=(6,6)$: then $\dreg= 6 < \ireg=7$.
\end{my_enumerate}

Only the case~\ref{item:1} is observed with generic homogeneous systems.
Furthermore, examples~\ref{item:1} and \ref{item:2} show that the degree of regularity depends upon the order of the variables (chosen in the description of the system of weights).
As the Hilbert series of a generic sequence doesn't depend on that order, it shows that we probably need to find a better tool in order to evaluate more precisely the degree of regularity in the quasi-homogeneous case.
However, the above bound already leads to good improvements on the complexity bounds, as we will see in the following sections.
Also note that these computations only hold when the system is
0-dimensional, we will discuss that restriction in
section~\ref{sec:noether-position}.

\vspace{-0.2cm}
\paragraph*{Genericity}
\label{sec:genericity}
We now prove that zero-dimensional $W$-homoge\-neous sequences of given $W$-degree are generically regular, under some assumptions on the $W$-degree.
Let us start with the first part of this statement:
\vspace{-0.2cm}
\begin{lemme}
  \label{lemme:reg_sequence_gen}
    Let $n$ be a positive integer, and consider the algebra $A := \K[X_{1},\dots,X_{n}]$, graded with respect to the system of weights $W = (w_{1},\dots,w_{n})$.
  Regular sequences
  of length $n$ % Ajout 14/01
  form a Zariski-open subset of all sequences of quasi-homogeneous polynomials of given $W$-degree in $A$.
\end{lemme}
\begin{proof}
  Let $(d_{1},\dots,d_{m})$ be a family of $W$-degrees, we consider the set $V\left(\K[\mathbf{a}][\mathbf{X}]\right)$ of all systems of quasi-homogeneous polynomials of $W$-degree $d_{1},\dots,d_{m}$, where $\mathbf{a}$ is a set of variables representing the coefficients of the polynomials.
  We denote by $f_{1},\dots,\hspace{-0.02em}f_{m}$ the polynomials of the generic system, and by $I$ the ideal they generate, in $\K[\mathbf{a}][\mathbf{X}]$.

  Since the Hilbert series~\eqref{eq:2} characterizes regular sequences (\cite[cor.~3.2]{Stanley78}), the sequence $(f_{i})$ is regular if and only if the ideal $I$ contains all monomials of $W$-degree between $\HI(I)+1$ and  $\HI(I)+ \max\{w_{i}\}$, where $\HI(I)$ is given by $\sum (d_{i} - w_{i})$.
  This expresses that a given set of linear equations has solutions, and so it can be coded as some determinants being non-zero.
\end{proof}
\vspace{-0.2cm}

There are some systems of $W$-degree for which there is no regular sequence.
The reason is that because of the weights, for some systems of $W$-degrees, there exists no or very few monomials.
For example, take $n=2$, $W=(1,2)$ and $\D=(1,1)$.
All quasi-homogeneous polynomials of $W$-degree $1$ are in $\K X$, so there is no regular sequence of quasi-homogeneous polynomials with these $W$-degrees.

However, if we only consider ``reasonable'' systems of $W$-degrees, that is systems of $W$-degrees for which there exists a regular sequence, regular sequences form a Zariski-dense subset from the above.

\vspace{-0.2cm}
\begin{remark}
  \label{rem:regularity1}
  A sufficient condition for example is to take weighted degrees such that
  $d_{1}$ is divisible by $w_{1}$, \ldots, $d_{n}$ is divisible by $w_{n}$.
  Thus we can define the sequence $X_{1}^{\sfrac{d_{1}}{w_{1}}},\dots,X_{n}^{\sfrac{d_{n}}{w_{n}}}$, which is regular, and so for such systems of weight, the regularity condition is generic.
\end{remark}
\vspace{-0.2cm}

We only proved the genericity for quasi-homogeneous sequences of length $n$, the more general case of a sequence of length $m \leq n$ will be proved in section~\ref{sec:noether-position} (remark~\ref{rem:noether1}).

\subsection{Noether position}
\label{sec:noether-position}
To compute the degree and degree of regularity of quasi-homoge\-neous
systems of positive dimension, we will %need to make more assumptions on the system, in order to specify which of the variables actually ``take part'' in the system.
%We will
assume that the
system $F=(f_{1},\dots,f_{m})$ we consider is in
\emph{Noether position} (as seen in \cite[ch.~13,~sec.~1]{eisenbud95} or \cite[def.~2]{Salvy2012}), i.e. the ideal $I = \langle F \rangle$ satisfies the two following conditions:
%that is that the canonical morphism $ \K[X_{n-r+1},\dots,X_{n}] \to \K[X_{1},\dots,X_{n}]/I$ is injective and integral.
%that is that the projection $\pi : V(I) \to V(X_{1},\dots,X_{m})$ onto the $n-m$ last coordinates is a surjective morphism with finite fibers.
% \begin{my_newenumerate}
% \item \label{item:H2}%
%there is
%$n_{m} > 0$ such that $x_{i}^{n_{i}} \in \LT(\langle f_{1}, \dots, f_{i} \rangle)$.
\begin{my_itemize}
  \item for $i \leq m$, the canonical image of $X_{i}$ in $\K[\mathbf{X}]/I$ is an algebraic integer over $\K[X_{m+1},\dots,X_{n}]$;
  \item $\K[X_{m+1},\dots,X_{n}] \cap I= 0$.
\end{my_itemize}
%Note that by induction, since the composition of two integral homomorphisms is an integral homomorphism, the first hypothesis is equivalent to saying that for any $i \leq m$, there is a $n_{i}$ such that $X_{i}^{n_{i}} \in \LT(I)$.

%An equivalent geometric definition can be found in \cite{MilneAG}.
%\myvspace{-0.3cm}  \item \label{item:H2}for any $1 \leq i \leq m$, the sequence $f_{1}, \dots, f_{i}$ is in Noether position for the variables $X_{1},\dots,X_{i}$;
%\myvspace{-0.3cm}  %\item \label{item:H2prime}for any $1 \leq i \leq m$, there exists $n_{i} > 0$ such that $x_{i}^{n_{i}} \in \LT(\langle f_{1}, \dots, f_{i} \rangle)$.\myvspace{-0.3cm}
%\end{my_newenumerate}
%
% that is that the sequence $f_{1},\dots,f_{m},X_{m+1},\dots,X_{n}$ is
% regular.  There are other possible definitions, for example in terms
% of commutative algebra (\cite[ch.~13,~sec.~1]{eisenbud95}) or of
% algebraic geometry (\cite{MilneAG}).  Noether position then refers to
% the case when the identity is a suitable isomorphism for Noether
% normalization theorem.
%
% Geometrically, an ideal in Noether position is an ideal such that the canonical projection of $V(I)$ over the $r$ first coordinates is a finite surjective morphism, that is an affine morphism such that any point of the target space has a non-zero finite number or reverse images.
%
%
% This property is indeed weaker than the property of being in Noether position (defined in \cite[ch.~13,~sec.~1]{eisenbud95} or \cite{MilneAG}).
%
%
\vspace{-0.2cm}
  \begin{lemme}
    \label{lemme:position-noether-reg}
    Let $F= f_{1},\dots,f_{m}$ be a regular quasi-homo\-geneous sequence of polynomials in $\K[X_{1},\dots,X_{n}]$.
    The sequence $F$ is in Noether position if and only if $\Fext := f_{1},\dots,f_{m},X_{m+1},\dots,X_{n}$ is a regular sequence. 
    \myvspace{-0.12cm}
  \end{lemme}
  \begin{proof}
    Let $I$ be the ideal generated by the $f_{i}$'s.
    The geometric characterization of Noether position (see e.g. \cite{MilneAG}) shows that the canonical projection onto the $m$ first coordinates
    \begin{equation}
    \pi : V(I) \to V(\langle X_{1}, \dots ,X_{m} \rangle)\label{eq:6}
    \end{equation}
    is a surjective morphism with finite fibers.
    This implies that the variety
    % $V(\langle f_{1}, \dots, f_{m}, X_{m+1}, \dots X_{n}\rangle)$,
    $V(\langle \Fext \rangle)$,
    that is $\pi^{-1}(0)$, is zero-dimensional, and so the sequence is regular.

    Conversely, assume $\Fext$ is a regular sequence.
    Let $i \leq m$, we want to show that $X_i$ is integral over the ring $\K[X_{m+1},\dots,X_{n}]$.
    Since $\Fext$ defines a zero-dimensional ideal, there exists $n_i \in \N$ such that $X_{i}^{n_{i}} = \LT(f)$ with
    %$f \in \langle f_{1},\dots,f_{m},X_{m+1},\dots,X_{n} \rangle$
    $f \in \langle \Fext \rangle$ 
    for the \grevlex ordering with $X_1 > \dots > X_n$.
    By definition of the \grevlex ordering, we can assume that $f$ simply belongs to $I$. %, which shows that $X_{i}$ is integral over $\K[X_{m+1},\dots,X_{n}]$ from the remark above.
    This shows that every $X_{i}$ is integral over $\K[X_{i+1},\dots,X_{n}]/I$. %, and we have the wanted result by induction, since the composition of integral homomorphisms is integral.
    We get the requested result by induction on $i$ : first, this is clear if $i=m$.
    Now assume that we know that $\K[X_{i},\dots,X_{n}]/I$ is an integral extension of $\K[X_{m+1},\dots,X_{n}]$.
    From the above, we also know that $X_{i-1}$ is integral over $\K[X_{i},\dots,X_{n}]$, and so, since the composition of integral homomorphisms is integral, we get the requested result.
    % Furthermore, the morphism $\K[X_{m+1},\dots,X_{n}] \to \K[\mathbf{X}]/I$ is injective, because its kernel is $I \cap \K[X_{\geq m+1}]$.
    
    Finally, we want to check the second part of the definition of Noether position.
    Assume that there is a non-zero polynomial in $\K[X_{m+1},\dots,X_{n}] \cap I$, since the ideal is quasi-homogeneous, we can assume this polynomial to be quasi-homogeneous.
    Either this polynomial is of degree 0, or it is a non-trivial syzygy between $X_{m+1},\dots,X_{n}$.
    So in any case, it contradicts the regularity hypothesis.
  \end{proof}
\vspace{-0.25cm}

  %
  %%% Preuve de la réciproque
  %%%%%%%%%%%%%%%%%%%%%%%%%%%
  % Assume $f_{1},\dots,f_{m},X_{m+1},\dots,X_{n}$ is a regular sequence.
  % We want to show that every $X_i$ for $i \leq m$ is integral over $\K[X_{m+1},\dots,X_{n}]$.
  % Since it defines a zero-dimensional ideal, for any $i$, there is a $n_i$ such that $X_{i}^{n_{i}} = \LT(f)$
  % with $f \in \langle f_{1},\dots,f_{m},X_{m+1},\dots,X_{n}$ for the \grevlex ordering with $X_1 > X_2 > \dots > X_n$.
  % By definition of the \grevlex ordering, we can assume that $f$ simply belongs to $I$.
  % This shows that every $X_{i}$ is integral over $\K[X_{i+1},\dots,X_{n}]$, and we have the wanted result by induction, since the composition of integral homomorphisms is integral.
  % And the morphism $\K[X_{m+1},\dots,X_{n}] -> \K[\mathbf{X}]/I$ is injective, because its kernel is $I \cap \K[X_{\geq m+1}]$.
  % Assume this kernel is nonempty, there is no chance the sequence $f,X_{\geq m+1}$ will be regular. (???)

As we did for regular sequences, we first show how we can evaluate the degree and degree of regularity of a sequence in Noether position, and then we show that the Noether position property is generic under some assumptions on the $W$-degree of the polynomials.
\vspace{-0.2cm}
\begin{theorem}
  Let $W=(w_{1},\dots,w_{n})$ be a system of weights, and $f_{1},\dots,f_{m}$ a regular sequence in Noether position, of quasi-homoge\-neous polynomials of $W$-degrees $(d_{1},\dots,d_{m})$.
  The same way we did above, we denote by $I$ % (resp. $\tilde{I}$)
  the ideal generated by the $f_{i}$'s% (resp. by the $\hom{W}(f_{i})$)
  .
  Then we have $\deg(I) = \prod_{i=1}^{m} \frac{d_{i}}{w_{i}}$ and
  $\dreg(I) \leq \sum_{i=1}^{m} \big(d_{i} - w_{i}\big) + \max\{w_{i}\}$.
  %: %\vskip-0.5cm
%   \begin{equation}
%     \label{eq:46}%\textstyle
%     \deg(I) = \prod_{i=1}^{m} \frac{d_{i}}{w_{i}} ; \;
% %    \textstyle
%     \dreg(I) \leq \sum_{i=1}^{m} \Big(d_{i} - w_{i}\Big) + \max\{w_{i}\}.
%   \end{equation}
  % \begin{gather}
  %   \label{eq:46}\textstyle
  %   \deg(I) = \prod_{i=1}^{m} \frac{d_{i}}{w_{i}} ; \\
  %   \label{eq:5}\textstyle
  %   \dreg(I) \leq \sum_{i=1}^{m} (d_{i} - w_{i}) + \max\{w_{i}\}.
  % \end{gather}
\end{theorem}%\myvspace{-0.7cm}
\begin{proof}
  Let us %write $I' := \langle f_{1},\dots,f_{m},X_{m+1},\dots,X_{n}\rangle$.
  denote by $I'$ the ideal generated by $F_{\mathrm{ext}}$.
%  the $f_{i}$'s ($i \leq m$) and the $X_{j}$'s ($m < j \leq n$).
  The degree of the ideal $I'$ is the same as that of  $I$, because the variety it defines is the intersection of $V(I)$ with some non-zero-divisor hyperplanes.
  Furthermore, all critical pairs appearing in a run of \F5 on $F$ will also appear in a run of \F5 on $F_{\mathrm{ext}}$, ensuring that $\dreg(F)\leq \dreg(F_{\mathrm{ext}})$.
  % Furthermore, if we let $G$ be a Gröbner basis of $I$, the family $G'$ defined by appending the polynomials $X_{m+1},\dots,X_{n}$ to $G$ is a Gröbner basis of $I'$.
  % This implies that $\dreg(I) \leq \dreg(I')$.

  But since by Noether position, the family $F_{\mathrm{ext}}$ defines a zero-dimensional variety, we can use the previous computations to deduce its degree of regularity and the degree of $I'$.\qed
\end{proof}
\begin{lemme}
  \label{lemme:position-noether-gen}
  Let $n$ be a positive integer, and consider the algebra $A := \K[X_{1},\dots,X_{n}]$, graded with respect to the system of weights $W = (w_{1},\dots,w_{n})$.
  Systems in Noether position form a Zariski-open subset of all systems of quasi-homogeneous polynomials of given $W$-degrees in $A$.
\end{lemme}
\begin{proof}
  Let $F=(f_{1},\dots,f_{m})$ be $m$ generic quasi-homoge\-neous
  polynomials, with coefficients in $\K[\mathbf{a}]$.  We use the same
  characterization of a zero-dimensional regular sequence as we did in
  the proof of Lemma~\ref{lemme:reg_sequence_gen}.  It allows us to
  express the regularity condition for the sequence
  $(f_{1},\dots,f_{m},X_{m+1},\dots,X_{n})$ as some determinants
  being non-zero, which by definition, shows that the condition of
  being in Noether position is an open condition.
\end{proof}

\vspace{-0.2cm}
Since a sequence in Noether position is in particular a regular sequence, we are confronted with the same problem as for the genericity of regular sequences, that is the possible emptiness of the condition.
However, it is still true that for ``reasonable'' systems of $W$-degrees, i.e. systems of $W$-degrees for which there exists enough monomials, sequences in Noether position do exist, and thus form a Zariski-dense subset of all sequences.
For example, since the sequence $X_{1}^{\sfrac{d_{1}}{w_{1}}},\dots,X_{m}^{\sfrac{d_{m}}{w_{m}}}$ is in Noether position, the sufficient condition given in Remark~\ref{rem:regularity1} is also sufficient to ensure that sequences in Noether position are Zariski-dense.

\vspace{-0.2cm}
\begin{remark}
  \label{rem:noether1}
  Any sequence in Noether position is in particular a regular sequence, so Lemma~\ref{lemme:position-noether-gen} proves that, under the same assumption on the degree, regular sequences of length $m \leq n$ are generic among quasi-homogeneous sequences of given $W$-degree.
\end{remark}
\vspace{-0.4cm}
\section{Computing Gröbner bases}% for \\quasi-homogeneous systems}
\label{sec:comp-grobn-basis}

\vspace{-0.3cm}
\subsection{Using the standard algorithms on the homogenized system}
\label{sec:using-standard-algorithms}

As we said before, in order to apply the \F5 algorithm to a quasi-homogeneous system, we may run it through $\hom{W}$.
This is shown by the following proposition.
\vspace{-0.2cm}
\begin{prop}
  \label{lemme:lemme:passage_par_jW}
  Let $F=(f_{1},\dots,f_{m})$ be a family of polynomials in $\K[X_{1},\dots,X_{n}]$, assumed to be quasi-homogeneous for a system of weights $W = (w_{1},\dots,w_{n})$.
  Let $\ordre1$ be a monomial order, $G$ the reduced Gröbner basis of $\hom{W}(F)$ for this order, and $\ordre2$ the pullback of $\ordre1$ through $\hom{W}$.
  Then%\myvspace{-0.3cm}
  \begin{my_enumerate}
    \item all elements of $G$ are in the image of $\hom{W}$;
    \item the family $G' := \hom[-1]{W}(G)$ is a reduced Gröbner basis of the system $F$ for the order $\ordre2$.
  \end{my_enumerate}
\end{prop}
\begin{proof}
  % For the first statement,
  The morphism $\hom{W}$ preserves $S$-polynomials, in the sense that
  $\Spol(\hom{W}(f),\hom{W}(g)) = \hom{W}(\Spol(f,g))$.
  Recall that we can compute a reduced Gröbner basis by running the Buchberger algorithm, which involves only multiplications, additions, tests of divisibility and computation of $S$-polynomials.
  Since all these operations are compatible with $\hom{W}$, if we run the Buchberger algorithm on both $F$ and $\hom{W}(F)$ simultaneously, they will follow exactly the same computations up to application of $\hom{W}$.
  The consequences on the final reduced Gröbner basis follow.
\end{proof}
\vspace{-0.2cm}

In practice, if we want to compute a \lex Gröbner basis of $F$, we generate the system $\tilde{F}=\hom{W}(F)$, we compute a $\grevlex$ basis
$\tilde{G_1}$ of $\tilde{F}$ with \F5, and then we compute a \lex Gröbner basis $\tilde{G_{2}}$ of $\tilde{F}$ with FGLM.
In the end, we get a \lex Gröbner basis of $\tilde{F}$, which we turn into a \lex Gröbner basis of $F$ via $\hom[-1]{W}$.

% We immediately notice that this method has several drawbacks
% for both Matrix-\F5 and FGLM: % Ajout 15/01/13
% we have seen, at least for regular sequences, that going through $\hom{W}$ increases the degree and the known bound for the degree of regularity of the system.
% Furthermore, for a given integer $d$, there are way more monomials of $\mathbf{1}$-degree $d$ than of $W$-degree $d$.
% This results in Matrix-\F5 working with matrices that are larger than necessary.

%\myvspace{-0.1cm}
\subsection{Direct algorithms} % TODO : Peut-être à reformuler
\label{sec:matrix-f5-algorithm}

% Going through $\hom{W}$ is not necessary for FGLM, we can avoid the impact of a greater degree of the ideal on the complexity estimate.
%Let us first consider the case of the change of ordering.
% We immediately notice that this method has several drawbacks
% for both Matrix-\F5 and FGLM: % Ajout 15/01/13
% we have seen, at least for regular sequences, that going through $\hom{W}$ increases the degree and the known bound for the degree of regularity of the system.
% Furthermore, for a given integer $d$, there are way more monomials of $\mathbf{1}$-degree $d$ than of $W$-degree $d$.
% This results in Matrix-\F5 working with matrices that are larger than necessary.
We can now explain why algorithm FGLM becomes a bottleneck with the above strategy.
Indeed, we have seen that going through $\hom{W}$ increases the Bézout bound of the system by a factor $\prod_{i=1}^{n}w_{i}$, and recall that the complexity of the FGLM step is polynomial in that bound.

Here is a workaround.
In the above process, we can apply $\hom[-1]{W}$ to the basis $\tilde{G_{1}}$ and thus obtain a \Wgrevlex basis $G_{1}$ of $F$.
We can then run FGLM on that basis to obtain a \lex basis of $F$.
Thus, we can avoid the problem of a greater degree of the ideal on the complexity of the FGLM step.

Algorithm $\F5$ operates by computing $S$-pairs, and as such, the argument of the proof of proposition~\ref{lemme:lemme:passage_par_jW} can be adapted, showing that going through $\hom{W}$ is equivalent to running a \F5 algorithm following weighted degree instead of total degree.
However, to evaluate the complexity of the \F5 algorithm, we instead study a less-efficient variant called Matrix-\F5 (described for example in \cite{FR09}), which needs to be adapted to the quasi-homogeneous case.
% About Matrix-\F5, we can also avoid going through $\hom{W}$,
All we need to do is change the algorithm a little, in order to consider directly the variables with their weight.
The modified algorithm is algorithm~\ref{algo:F5matqh} opposite.
The function \textsc{F5Criterion}$(\mu,i,\mathcal{M})$ implements the \F5-criterion described in~\cite{Fau02a}: it evaluates to false if and only if $\mu$ is the leading term of a line of the matrix $\mathcal{M}_{d-d_{i},i-1}$.
The function \textsc{EchelonForm}$(M)$ reduces the matrix $M$ to row-echelon form, not allowing any row swap.

\begin{algorithm}[h]
  \caption{Matrix-\F5 ($W$-homogeneous version)}
  \label{algo:F5matqh}
  \KwIn{$
  \begin{cases}
    f_{1},\dots,f_{m} \text{ $W$-homogeneous polynomials}\\
    \quad\quad\quad\text{with $W$-degrees $d_{1},\dots,d_{m}$}\\
    d_{\text{max}} \in \N
  \end{cases}$}
\KwOut{$G$ Gröbner basis of $\langle f_{1},\dots,f_{m} \rangle$ up to \hbox{$W$-degree} $d_{\text{max}}$}
$G \leftarrow \{f_{1},\dots,f_{m}\}$ \;
\For{$d = 1$ \KwTo $d_\text{max}$}{
  $\mathcal{M}_{d,0} \leftarrow$ matrix with 0 lines\;
  \For{$i = 1$ \KwTo $m$}{
    \If{$d = d_{i}$}{
      $\mathcal{M}_{d,i} \leftarrow \tilde{\mathcal{M}_{d,i-1}} \append$ line $f_{i}$ with label $(1,f_{i})$\;
    }
    \ElseIf{$d > d_{i}$}{
      $\mathcal{M}_{d,i} \leftarrow \tilde{\mathcal{M}_{d,i-1}}$\;
      \For{$j = 1$ \KwTo $n$}{
        \ForAll{lines $f$ of $\tilde{\mathcal{M}_{d-w_{j},i}}$ with label $(e,f_{i})$ s.t. the biggest variable dividing $e$ is $x_{j}$}{
          \For{$k = n$ \KwDownTo $j$}{
            \If{\textsc{F5Criterion}$(x_{k}e,i,\mathcal{M})$}{
              $\mathcal{M}_{d,i} \leftarrow \mathcal{M}_{d,i} \append x_{k}f$ with label $(x_{k}e,f_{i})$\;
            }
          }
      }
    }
    }
  }
 $\tilde{\mathcal{M}_{d,m}} \leftarrow \text{\textsc{EchelonForm}}(\mathcal{M}_{d,m})$ \;
 For any line having been reduced to a non-zero polynomial, append it to $G$ \;
}
\Return{$G$}
\end{algorithm}

% The only part that differs from the classical Matrix-\F5 is the loop, from line~\ref{algoline:F5mat_qh_mod1} to line~\ref{algoline:F5mat_qh_mod2}.
% Recall that the classical algorithm loops on the labels $(e,f_{i})$ in degree $d-1$, and check the \F5 criterion on all new monomials of degree $\left(d-\deg(f_i)\right)$ one can obtain from $e$.
% Instead, this algorithm loops on monomials of $W$-degree $\left(d-\deg(f_i)\right)$, checks the \F5 criterion on these monomials, and then only checks if there's a label $(e,f_{i})$ (possibly in $W$-degree lower than $d-1$) which might lead to that monomial.
\vspace{-0.2cm}
\subsection{First complexity bounds}
\label{sec:first-compl-bounds}
Let $F=(f_{1},\dots,f_{n})$ be a system of $W$-homogeneous
polynomials in $\K[X_{1},\dots,X_{n}]$, and let $I$ be the ideal
generated by $F$, $D$ the degree of $I$, $\dreg$ its degree of regularity
and $\HI$ its index of regularity.  %We write $\tilde{F} = \hom{W}(F)$,
%and $\tilde{I}$, $\tilde{D}$ 
%, $\tilde{\dreg}$
%and $\tilde{\HI}$
%respectively its generated ideal, degree %, degree of regularity
%and index of regularity.
The classical complexity bounds of Matrix-\F5 (for a regular system) and FGLM are
\begin{equation}
  C_{F_{5}} = O\left(\dreg M_{\dreg,W}(n)^{\omega}\right) ;\;
  C_{FGLM} = O\left(n D^{3}\right),
\end{equation}
%where $\dreg$ and $D$ are respectively used to denote the degree of regularity and the degree of the ideal $\langle F\rangle$ %(instead of $\langle \tilde{F}\rangle$)
where $M_{d,W}(n)$ stands for the number of monomials of $W$-degree $d$ in $n$ variables (see for example \cite{Bar04} for \F5 and \cite{FGLM} for FGLM).

Assuming the system $F$ is a regular sequence, we have already seen
the following estimates: % between the parameters :
\begin{equation}
\dreg \leq \ireg + \max\{w_{i}\}\;;\hspace{0.5cm}
D = \frac{\prod_{i=1}^{n}{d_{i}}}{\prod_{i=1}^{n} w_{i}}.\label{eq:12}
\end{equation}
% \begin{gather}
%   \label{eq:58}
%   % \begin{aligned}
%   %   \dreg = \tilde{\dreg} & \leq \ireg + \max\{w_{i}\} \\ &< \tilde{\ireg} - \sum_{i=1}^{n}w_{i} + n -1 + \max\{w_{i}\}; 
%   % \end{aligned} \\
%   \textstyle \dreg \leq \tilde{\ireg} - \left(\sum_{i=1}^{n}w_{i} - n\right) + \left(\max\{w_{i}\} - 1\right); \\
%   \label{eq:59}
%   D = \frac{\tilde{D}}{\prod_{i=1}^{n} w_{i}}.
% \end{gather}
If we compare these values with their equivalent with the system of weights $\mathbf{1}$, we notice a significant gain in theoretical complexity bounds for both the FGLM and \F5 algorithms.

But this gain in complexity for \F5 does not take into account the size of the computed matrices.
That size is necessarily reduced, because the number of monomials of given $W$-degree is much smaller than the number of monomials of given \hbox{$\mathbf{1}$-degree}.
The point of the following lemma is to evaluate this gain.\myvspace{-0.2cm}
\begin{lemme}
  \label{lemme:nombre_monomes_W}
  Let $W = (w_{1},\dots,w_{n})$ be a system of weights, and for any $i$, let $W_{i} = (w_{1},\dots,w_{i})$.
  For any integer $d$, we denote by $M_{d,W}(n)$ the number of monomials of $W$-degree $d$, that is the size of the matrix of $W$-degree $d$.
  Let $\delta := \pgcdfam{W}$, $P := \prod_{i=1}^{n}w_{i}$, $S_{i}$ the integer defined recursively as following:
    \begin{gather}
    \label{eq:119}
    S_{1} = 0,\; %\\
%    \label{eq:125}
    S_{i} = S_{i-1} + w_{i} \cdot\frac{\pgcdfam{W_{i-1}}}{\pgcdfam{W_{i}}} \text{ for } i \geq 2 
  \end{gather}
  and $T_{i}$ the integer defined recursively as following:
  \begin{equation}
    \label{eq:127}
    T_{1} = 0,\;% \\
    T_{i} = T_{i-1} + w_{i} \cdot\left(\frac{\pgcdfam{W_{i-1}}}{\pgcdfam{W_{i}}}-1\right) -1 \text{ for } i \geq 2.
  \end{equation}
  Then the number of monomials of $W$-degree $d$ is bounded above and
  below by:
%$$ Displaymath c'est le maaaaal
  \begin{equation}
  \frac{\delta}{P}M_{d-T_{n}-n+1,\mathbf{1}}(n)
    \leq M_{d,W}(n)\leq     \frac{\delta}{P} M_{d+ S_{n} - n +1,\mathbf{1}}(n).\label{eq:10}
  \end{equation}
%$$
  % \begin{multline}
  %   \label{eq:126}
  %   \frac{\delta}{P}M_{d-T_{n}-n+1,\mathbf{1}}(n)
  %   \leq M_{d,W}(n) \\ \leq
  %   % \frac{\delta}{P} \cdot \binom{d + S_{n}}{n-1} =
  %   \frac{\delta}{P} M_{d+ S_{n} - n +1,\mathbf{1}}(n).
  % \end{multline}  
\end{lemme}
\begin{proof}
  This is a consequence of theorems 3.3 and 3.4 in~\cite{Geir02}, if we recall that
  \mbox{$M_{d,\mathbf{1}}(n) = \binom{d+n-1}{d} = \binom{d+n-1}{n-1}$}.
\end{proof}%\vspace{-0.25cm}
% The table~\ref{tab:bornes_monomes1} shows sample values of these bounds.
% The last column of the table is presented for comparison purpose, as it is the size of the matrices the algorithm Matrix-\F5 works with when considering the system through $\hom{W}$.

Note that if $W = \mathbf{1}$, the bounds we get are trivial, which means the complexity bounds we will obtain with them will specialize without any difficulty to the known bounds for the homogeneous case.

% Tableau des valeurs des bornes
% \input{./tex/Bounds_values}

Using the notation $S=S_{n}$% and $\Sigma=\sum_{i=1}^{n}w_{i}$
, we get this new complexity bound for quasi-homogeneous Matrix-\F5:
\begin{align}
  \label{eq:69}
  C_{F_{5}} &= O\left(\dreg M_{\dreg,W}(n)^{\omega}\right) \\
  % \label{eq:70}
  % &\leq O\left(\left(\tilde{\HI} - \Sigma-1 + \max\{w_{i}\}\right) \left[\frac{1}{P} M_{\tilde{\HI} - \Sigma +n + S-n -1 + \max\{w_{i}\},\mathbf{1}}(n)\right]^{\omega}\right) \\
  \label{eq:71}
  & \begin{aligned}
    = O&\left(\vphantom{\binom{S}{S}^{\omega}}\Big(
      %\tilde{\HI} - \Sigma-1 + \max\{w_{i}\}
      \ireg + \max\{w_{i}\}
     \Big)\right.\\[-0.2cm]
     & \,\,\;\left. \cdot \left[\frac{\delta}{P}
       \binom{\ireg + \max\{w_{i}\} + S - 1}{n-1}
       %M_{\tilde{\ireg} -1 + \max\{w_{i}\}  +S-\Sigma,\mathbf{1}}(n)
      \right]^{\omega}\right).
  \end{aligned}
\end{align}
On the other hand, the estimate on the degree of a quasi-homoge\-neous variety gives the following complexity bound for FGLM:
\begin{equation}
  \textstyle C_{FGLM} = O\left(n \left[\frac{\tilde{D}}{P}\right]^{3}\right),\label{eq:9}
  \end{equation}
where $\tilde{D} = \prod_{i=1}^{n} d_{i}$ is the degree of the ideal $\langle \hom{W}(F) \rangle$.
In the end, for the whole process, we can see that 
the complexity bound for our direct strategy is smaller by a factor of $P^{\omega}$, when compared to the strategy of going through $\hom{W}$.

% Remarque sur le cas de dimension positive en position de Noether
% \input{./tex/Positive_dimension}

% A consequence of these results is that in general, the higher the weights of the variables, the bigger the gain.
% In particular, it shows why the applying Matrix-\F5 on $\hom{W}(F)$ did not yield satisfactory complexity results.

% On the other hand, it also shows that when possible, replacing powers of the variables with higher weight variables allows for faster computations of a Gröbner basis.
% % That's the idea behind algorithm~\ref{algo:AlgoAvecPreproc}, in the case of an any-dimensional system for which we want a \Wgrevlex basis.
% % One can of course proceed similarly in order to compute a basis for a 0-dimensional system, via FGLM.
% % \myvspace{-0.3cm}\begin{algorithm}
% %   \caption{Matrix-\F5 with preprocessing to increase the weights}
% %   \label{algo:AlgoAvecPreproc}
% %   \input{./algo/Weight_preprocessing}
% % \end{algorithm}\myvspace{-0.3cm}
% Of course, on a generic homogeneous or quasi-homogeneous system, it won't be possible,
% but since in practice, the systems we consider are
% not generic, situations may appear where this idea proves interesting.
% We will see an example of that in section~\ref{sec:cas-affine}.

%\vspace{-0.2cm}
\subsection{Precise analysis of matrix-{\large\F5}}
\label{sec:compl-fine-fcinq}

% TODO: Reformuler
Let us now follow more closely the computations occurring in the Matrix-\F5 algorithm, and obtain more accurate complexity bounds.
For this purpose, we take on the computations made in~\cite[ch.~3]{Bar04}, without proving them whenever the proof is an exact transcription of the homogeneous case.

Let $W= (w_{1},\dots,w_{n})$ be a system of weights, and $f_{1},\dots, f_{m}$ a system of quasi-homogeneous polynomials in $\K[X_{1},\dots,X_{n}]$, which we assume satisfies the hypotheses \ref{item:H1} and \ref{item:H2}.
% \begin{my_newenumerate}%[align=left,label=\textbf{\hypCounter*.},ref=\textbf{\hypCounter*}]
%   \myvspace{-0.3cm} 
% \item \label{item:H1}the sequence $f_{1}, \dots, f_{m}$ is regular;
%   \myvspace{-0.3cm}
% \item \label{item:H2}for any $1 \leq i \leq m$, there exists $n_{i} >
%   0$ such that $x_{i}^{n_{i}} \in \LT(\langle f_{1}, \dots, f_{i}
%   \rangle)$ (in other words, $f_{1}, \dots, f_{i}$ is in Noether
%   position w.r.t $X_{1},\dots,X_{i}$ for $1\leq i \leq m$).
% %\myvspace{-0.3cm}  \item \label{item:H2}for any $1 \leq i \leq m$, the sequence $f_{1}, \dots, f_{i}$ is in Noether position for the variables $X_{1},\dots,X_{i}$;
% \myvspace{-0.3cm}  %\item \label{item:H2prime}for any $1 \leq i \leq m$, there exists $n_{i} > 0$ such that $x_{i}^{n_{i}} \in \LT(\langle f_{1}, \dots, f_{i} \rangle)$.\myvspace{-0.3cm}
%\end{my_newenumerate}
We denote by $(d_{1},\dots,d_{m})$ the respective $W$-degrees of the
polynomials $f_{1},\dots,f_{m}$, and we will assume them to allow the
existence of such systems.
% The hypotheses~\ref{item:H2} and \ref{item:H2prime} are equivalent (from the algebraic definition of Noether position) 
% and generic (from Lemma~\ref{lemme:position-noether-gen}).
% \todo{Reprendre ce passage une fois la définition de PN changée (supprimer l'une des hypothèses? les fusionner? laisser tel-quel et reformuler la phrase?)}
% Recall that by definition, it means that the sequence $f_{1},\dots,f_{i},X_{i+1},\dots,X_{n}$ is regular for any $1 \leq i \leq m$.

We also denote by:
\begin{my_itemize}
  \item $A_{i} = \K[X_{1},\dots,X_{i}]$, and $A = A_{n}$;
  \item $S_{i}$ the integer defined in Lemma~\ref{lemme:nombre_monomes_W}, and $S = S_{n}$;
%  \item $\Sigma_{i}=\sum_{j=1}^{i} w_{j}$, and $\Sigma = \Sigma_{n}$; % Inutile à présent (21/01)
  \item $P_{i} = \prod_{j=1}^{i} w_{j}$, and $P = P_{n}$;
  \item $I_{i} = \langle f_{1}, \dots, f_{i}\rangle$, and $I = I_{m}$;
  \item $\tilde{f_{j}} = \hom{W}(f_{j})$;
  \item $\tilde{I_{i}} = \langle \tilde{f_{1}}, \dots, \tilde{f_{i}}\rangle$, and $\tilde{I} = \tilde{I_{m}}$;
  \item $D_{i} = \deg(I_{i}) = \prod_{j=1}^{i} \left(d_{j}/w_{j}\right)$;
  \item $\tilde{D_{i}} = \deg(\tilde{I_{i}}) = \prod_{j=1}^{i} d_{j}$;
  \item $\dreg^{(i)}$ % = \tilde{\dreg^{(i)}}$
  the degree of regularity of $I_{i}$ (or of $\tilde{I_{i}}$)
  %: $\dreg^{(i)} \leq \sum_{j=1}^{i}(d_{j}-w_{j}) + \max\{w_{j}\}$
  ;
  % \item $\tilde{\dreg^{(i)}}$ the degree of regularity of $\tilde{I_{i}}$
  %: $\tilde{\dreg^{(i)}} \leq \sum_{j=1}^{i}(d_{j} - 1) +1$
  % ;
  % \item $\ireg^{(i)} = \ireg(I_{i}) = \sum_{j=1}^{i}(d_{j}-w_{j}) + \max\{w_{j}\}$  ;
  % \item $\tilde{\ireg^{(i)}} = \ireg(\tilde{I_{i}})\sum_{j=1}^{i}(d_{j} - 1) +1$  ;
  \item $G_{i}$ the \Wgrevlex Gröbner basis of $I_{i}$ as given by Matrix-\F5.
\end{my_itemize}
With these notations, we are going to prove the following theorem:%\myvspace{-0.15cm}
\begin{theorem}
  \label{thm:comp_fine}
  Let $W=(w_{1},\dots,w_{n})$ be a system of weights, and $f_{1},\dots,f_{m}$ ($m \leq n$) a system of $W$-homogeneous polynomials satisfying~\ref{item:H1} and~\ref{item:H2}.
  Then the complexity of quasi-homogeneous Matrix-\F5 algorithm (algorithm~\ref{algo:F5matqh}) is:
  \begin{equation}
    \label{eq:35}
    C_{\F5}  = O
    \left(
     \sum_{i=2}^{m} (D_{i-1}-D_{i-2}) M_{\dreg^{(i)},W}(i) M_{\dreg^{(i)},W}(n)
    \right)
  \end{equation}
\end{theorem}%\vspace{-0.1cm}

We aim at computing precisely how many lines are reduced in a run of the Matrix-\F5 algorithm, that is, the number of polynomials in the returned Gröbner basis.
This is done by the following proposition, which is a weak variant of \cite[th.~10]{Salvy2012}:%\vspace{-0.15cm}
\begin{prop}
  Let $(f_{1},\dots,f_{m})$ be a $W$-homogeneous system (w.r.t a system of weights $W$) satisfying the hypotheses~\ref{item:H1} and \ref{item:H2}.
  Let $G_{i}$ be a reduced Gröbner basis of $(f_{1},\dots,f_{i})$ for the \Wgrevlex monomial ordering, for $1 \leq i \leq m$.
  Then the number of polynomials of $W$-degree $d$ in $G_{i}$ whose leading term does not belong to $\LT(G_{i-1})$ is bounded by $b_{d,i}$, defined by the generating series
  \begin{equation}
    \label{eq:1}
    B_{i}(z) = \sum_{d=0}^{\infty} b_{d,i} z^{d} = z^{d_{i}} \prod_{k=1}^{i-1} \frac{1-z^{d_{k}}}{1-z^{w_{k}}}.
  \end{equation}
\end{prop}
\begin{proof}
  The proof of \cite[th.~10]{Salvy2012} still holds in the quasi-homoge\-neous case, using formula~\eqref{eq:2} for the Hilbert series of a quasi-homogeneous regular sequence.
  %using Lemmas~\ref{lemme:comp-fine-LT} and~\ref{lemme:comp_fine_etiq} where needed.
\end{proof}

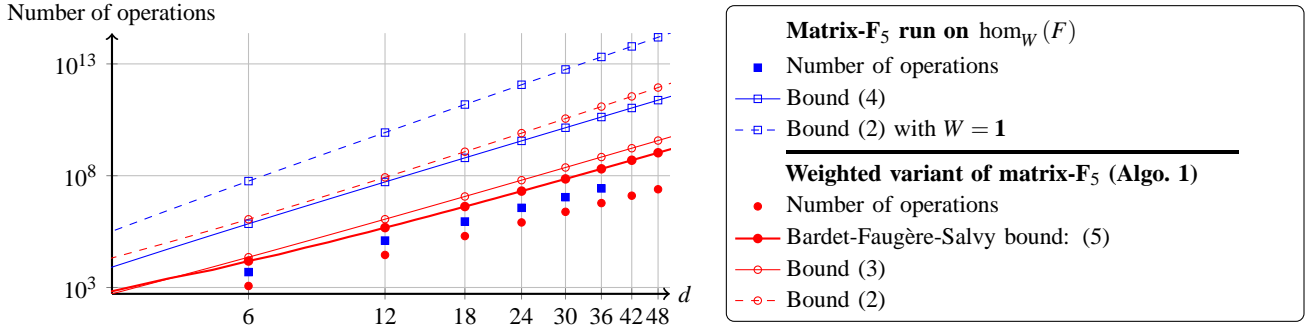
\begin{figure*}
  \centering
  \begin{tikzpicture}
    \begin{loglogaxis}
      [
      xlabel=$d$,
      xlabel style={
        at={(1,0)},
        anchor=west },
      xmin=3,
      xmax=51,
      ylabel={Number of operations},
      ylabel style={
        rotate=-90,
        at={(0,1)},
        anchor=south },
      xtick={0,6,...,48},
      xticklabel style={/pgf/number format/fixed},
      xticklabel={%
        \pgfmathfloatparsenumber{\tick}%
        \pgfmathfloatexp{\pgfmathresult}%
        \pgfmathprintnumber{\pgfmathresult}%
      },
      width=9cm,
      height=5.05cm,
      scaled ticks=false,
      mark options={ solid, scale=0.7},
      axis x line=bottom,
      axis y line=left,
      axis line style={
        ->,
        thick },
      legend style={
        at={(1.1,0.5)},
        anchor=west,
        nodes={text width=6.7cm,text depth={}},
        cells={anchor=west},
        rounded corners=3pt },
      grid=major,
      log basis y=10]
      
      \addlegendimage{empty legend}
      \addlegendentry{\textbf{Matrix-\F5 run on $\hom{W}(F)$}}
      
      \addplot[ only marks, mark=square*, blue]
      table [ x index=0, y index=4]
      {valeurs_complexite_1.out};
      \addlegendentry{Number of operations};
      
      \addplot[ mark=square, blue, solid, mark phase=4, mark repeat=6]
      table [ x index=0, y index=3] {bornes_complexites1.out};
      \addlegendentry{Bound \eqref{eq:89}};
      
      \addplot[ mark=square, blue, dashed, mark phase=4, mark repeat=6]
      table [x index=0, y index=4] {bornes_complexites1.out};
      \addlegendentry{Bound \eqref{eq:71} with $W=\mathbf{1}$};

      \addlegendimage{empty legend}
      \addlegendentry{\rule{6cm}{1pt}}
      \addlegendimage{empty legend}
      \addlegendentry{\textbf{Weighted variant of \hbox{matrix-\F5} (Algo.~\ref{algo:F5matqh})}}

      \addplot[ only marks, mark=*, red]
      table [x index=0, y index=1] {valeurs_complexite_2.out};
      \addlegendentry{Number of operations};
      
      \addplot[ mark=*, red, solid, thick, mark phase=4, mark repeat=6]
      table [ x index=0, y index=1] {borne_bardet.out};
      \addlegendentry{Bardet-Faugère-Salvy bound: \eqref{eq:4}};
      % Si on enlève cette courbe, penser à enlever aussi sa référence dans le texte

      \addplot[ mark=o, red, solid, mark phase=4, mark repeat=6]
      table [x index=0, y index=1] {bornes_complexites1.out};
      \addlegendentry{Bound \eqref{eq:86}};
      
      \addplot[ dashed, mark=o, red, mark phase=4, mark repeat=6]
      table [x index=0, y index=2] {bornes_complexites1.out};
      \addlegendentry{Bound \eqref{eq:71}};

    \end{loglogaxis}
  \end{tikzpicture}

  \caption{Bounds and values, on a log-log scale, for the number of arithmetic operations performed in Matrix-\F5 for a generic system with $W=(1,2,3)$ and $\D=(d,d,d)$}
  \label{fig:boundslog}
\end{figure*}
% }
%%\vspace{-0.3cm}
So we can obtain a better bound for the number of elementary operations performed in a Matrix-\F5 run.
Indeed, $B_{i}(1)$ represents the number of reduced polynomials in the computation of a Gröbner basis of $(f_{1},\dots,f_{i},X_{i+1},\dots,X_{n})$, that is as many as in the computation of a Gröbner basis of $(f_{1},\dots,f_{i})$:
since we only perform reductions under the pivot line,
% Lemma~\ref{lemme:comp_fine_etiq}
\cite[prop.~9]{Salvy2012} 
shows that the lines coming from $X_{i+1},\dots,X_{n}$ will not add any reduction.
Note that the above generating series is the same as the Hilbert series of $\langle f_{1}, \dots, f_{i-1}, X_{i}, \dots, X_{n} \rangle$, and so, that its value at $z=1$ is the degree of that ideal, or $D_{i-1}$.
Therefore, we know that the number of reduced polynomials with label $(m,f_{i})$ will be $D_{i-1}-D_{i-2}$ (with convention that $D_{0}=0$).

Now, let $g$ be any polynomial of $W$-degree $d$ being reduced in a run of the Matrix-\F5 algorithm on $(f_{1},\dots,f_{i})$.
From~\cite[prop.~9]{Salvy2012}, 
%Lemma~\ref{lemme:comp-fine-LT},
we know that the leading term of $g$, after reduction, is in $A_{i}$.
So overall, in $W$-degree $d$, we reduce by at most as many lines as there are monomials in $A_{i}$, that is $M_{d,W}(i)$.
Furthermore, each reduction costs at most $O(M_{d,W}(n))$ elementary algebraic operations, since this is the length of the matrix lines.
And we perform these reductions up to degree $\dreg^{{(i)}}$.
Note that, if $i=1$, there clearly isn't any reduction in the computation, and we obtain the following formulas:%\myvspace{-0.3cm}
\begin{align}
  % \label{eq:85}
  \label{eq:86}
  \hspace{-0.2cm} C_{\text{\F5}}
  % & = O
  % \left(
  %  \sum_{i=2}^{m} G_{i}(1) M_{\dreg^{(i)},W}(i) M_{\dreg^{(i)},W}(n)
  % \right) \\
  & = O
  \left(
   \sum_{i=2}^{m} (D_{i-1}-D_{i-2}) M_{\dreg^{(i)},W}(i) M_{\dreg^{(i)},W}(n)\hspace{-1.3mm}
  \right)\\
  % \label{eq:87}
  % & = O
  % \left(
  %  \sum_{i=2}^{m}
  %  {
  %    \frac{\tilde{D_{i-1}}}{P_{i-1}}
  %    M_{\tilde{\dreg^{(i)}}-\Sigma_{i}-1 + \max\{w_{j}\},W}(i)
  %    M_{\tilde{\dreg^{(i)}}-\Sigma_{i}-1 + \max\{w_{j}\},W}(n)
  %  }
  % \right) \\
  \label{eq:88}
  & \begin{aligned}
    = O \left(\vphantom{\sum_{i=2}^{m}} \right.& \sum_{i=2}^{m} \frac{1}{P_{i}P_{n}}\left(\frac{\tilde{D_{i-1}}}{P_{i-1}}-\frac{\tilde{D_{i-2}}}{P_{i-2}}\right)
    \\[-0.2cm]
    &\cdot M_{\dreg^{(i)} + S_{i} -i+1,\mathbf{1}}(i)  %\\
     % &
    \left. \hspace{-0.3em}
      \vphantom{\sum_{i=1}^{m}}\cdot M_{\dreg^{(i)} + S_{n} -n +1,\mathbf{1}}(n) \right) %\qed
  \end{aligned}
    % \label{eq:88}
  % &   \begin{aligned}
  %   = O \left( \sum_{i=2}^{m}\right. &\frac{\tilde{D_{i-1}}}{P_{i-1}P_{i}P_{n}} \\
  %   & \cdot M_{\tilde{\dreg^{(i)}} + S_{i} - \Sigma_{i} -1 + \max\{w_{j}\},\mathbf{1}}(i) \\
  %   & \left. \cdot M_{\tilde{\dreg^{(i)}} + S_{i} - \Sigma_{i} -1 + \max\{w_{j}\},\mathbf{1}}(n) \right)
  %  \end{aligned}
\end{align}%\vskip-0.5cm
%\paragraph{Remarks}

In comparison, the above reasoning for Matrix-\F5 applied to $\tilde{F}$ would give
\begin{equation}
  \label{eq:89}
  C_{\text{\F5}} = O
  \left(
   \sum_{i=2}^{m} \left(\tilde{D_{i-1}}-\tilde{D_{i-2}}\right) M_{\tilde{\dreg^{(i)}},\mathbf{1}}(i) M_{\tilde{\dreg^{(i)}},\mathbf{1}}(n)
  \right)
\end{equation}
so that here again, working with quasi-homogeneous polynomials yields a gain or roughly $P^{3}$.
Note that the exponent~$3$ (instead of the previous $\omega$) is not really meaningful, because we assumed here that we were using the naive pivot algorithm to perform the Gauss reduction.
However, if we assume $\omega = 3$ in the previous computations as well, we observe that our new bound is generally much better than the previous one: figure~\ref{fig:boundslog} shows a plot of data obtained both with algorithm~\ref{algo:F5matqh} and with Matrix-\F5 through $\hom{W}$, together with the different bounds we can compute.

Asymptotically, though, the gain does not look important, since the complexity is still $O(nD^{3})$ where $D$ is the degree of the ideal and $n \geq m$ the number of variables, or in $O(nd^{3n})$ where $d$ is the greatest $d_{i}$.
\vspace{-0.2cm}
\begin{remark}
  One may also push the computations a bit further, and obtain an even more accurate bound, expressed in terms of the $b_{d,i}$ (these calculations are done in \cite{Bar04} for the homogeneous case, and can easily be transposed to the quasi-homogeneous case):
  % \myvspace{-0.5cm}
  % \begin{equation}
  %   \label{eq:4} 
  %   \begin{split}
  %     C_{\text{\F5}} = O \left(
  %      \sum_{i=1}^{m-1} \sum_{d=0}^{\infty}\right.& \frac{b_{d+d_{i+1},i+1}}{P_{i+1}P_{n}} \\
  %     &\cdot M_{d + d_{i+1} + S_{i} - i ,\mathbf{1}}(i+1) \\
  %     &\left. \vphantom{\sum_{i=1}^{m-1}} \cdot M_{d + d_{i+1} + S_{n} -n +1,\mathbf{1}}(n) \right).
  %   \end{split}
  % \end{equation}
    \begin{multline}
    \label{eq:4} 
      \hspace{-0.5cm}C_{\text{\F5}} = O \left(
       \sum_{i=1}^{m-1} \sum_{d=0}^{\infty}\right. \frac{b_{d+d_{i+1},i+1}}{P_{i+1}P_{n}} \cdot M_{d + d_{i+1} + S_{i+1} - i ,\mathbf{1}}(i+1)\\[-0.2cm]
      \left.
       \vphantom{\sum_{i=1}^{m-1}} \cdot M_{d + d_{i+1} + S_{n} -n +1,\mathbf{1}}(n) \right).
  \end{multline}
  As an example, we computed that bound as well for a particular case, and included it in figure~\ref{fig:boundslog}.
  As one can see, that bound is indeed better than the intermediate evaluation~\eqref{eq:86}, but the difference is low enough to justify using the latter evaluation.
  Furthermore, the bound~\eqref{eq:86} expressed in terms of the $D_{i}$'s is more useful in practice, since it has a closed formula using only the parameters of the system ($n$, $m$, $d_{i}$ and $w_{i}$).
  That allows us to use it in complexity evaluations, in both theory and practice.
\end{remark}
\vspace{-0.1cm}
\vspace{-0.3cm}
\begin{remark}
  As one can see on figure~\ref{fig:boundslog}, the number of operations needed by Matrix-\F5 on the homogenized system is not significantly higher than the number of operations needed by the quasi-homogeneous variant of Matrix-\F5.
  That is mostly true because the unmodified algorithm can make use of some of the structure of the quasi-homogeneous systems (for example, columns of zeroes in the matrices).
\end{remark}
%\vspace{-0.2cm}

\begin{table*}[t]
  \centering
{\small
  \subfloat[Benchmarks with FGb]{%
  \begin{tabular}{lrllm{1.3cm}llm{1.3cm}}
    \hline
    System & $\deg(I)$ & $t_{\F5}$ (qh) & $t_{\F5}$ (std) & Speed-up for \F5 & $t_{\mathrm{FGLM}}$ (qh) & $t_{\mathrm{FGLM}}$ (std) & Speed-up for FGLM \\
    % BEGIN RECEIVE ORGTBL benchFGb2
\hline
Generic $n=7$, $W=(1^{4},2^{3})$, $\mathbf{D}=(4^{7})$ & \num{2048} & \SI{2.7}{\second} & \SI{3.4}{\second} & \num{1.2} & \SI{0.4}{\second} & \SI{1.1}{\second} & \num{2.6} \\
Generic $n=8$, $W=(1^{4},2^{4})$, $\mathbf{D}=(4^{8})$ & \num{4096} & \SI{12.3}{\second} & \SI{22.5}{\second} & \num{1.8} & \SI{2.4}{\second} & \SI{7.3}{\second} & \num{3.0} \\
Generic $n=9$, $W=(1^{5},2^{4})$, $\mathbf{D}=(4^{9})$ & \num{16384} & \SI{314.9}{\second} & \SI{778.5}{\second} & \num{2.5} & \SI{119.6}{\second} & \SI{327.8}{\second} & \num{2.7} \\
\hline
Generic $n=7$, $W=(2^{5},1^2)$, $\mathbf{D}=(4^{7})$ & \num{512} & \SI{0.1}{\second} & \SI{0.3}{\second} & \num{3.2} & \SI{0.1}{\second} & \SI{0.1}{\second} & \num{1.7} \\
Generic $n=8$, $W=(2^{6},1^2)$, $\mathbf{D}=(4^{8})$ & \num{1024} & \SI{0.4}{\second} & \SI{1.6}{\second} & \num{4.2} & \SI{0.2}{\second} & \SI{0.3}{\second} & \num{1.9} \\
Generic $n=9$, $W=(2^{7},1^2)$, $\mathbf{D}=(4^{9})$ & \num{2048} & \SI{1.6}{\second} & \SI{8}{\second} & \num{4.9} & \SI{0.6}{\second} & \SI{1.2}{\second} & \num{2.0} \\
Generic $n=10$, $W=(2^{8},1^2)$, $\mathbf{D}=(4^{10})$ & \num{4096} & \SI{7.5}{\second} & \SI{40.4}{\second} & \num{5.4} & \SI{2.4}{\second} & \SI{6.2}{\second} & \num{2.6} \\
Generic $n=11$, $W=(2^{9},1^2)$, $\mathbf{D}=(4^{11})$ & \num{8192} & \SI{33.3}{\second} & \SI{213.5}{\second} & \num{6.4} & \SI{17.5}{\second} & \SI{41.2}{\second} & \num{2.4} \\
Generic $n=12$, $W=(2^{10},1^2)$, $\mathbf{D}=(4^{12})$ & \num{16384} & \SI{167.9}{\second} & \SI{1135.6}{\second} & \num{6.8} & \SI{115.8}{\second} & \SI{246.7}{\second} & \num{2.1} \\
Generic $n=13$, $W=(2^{11},1^2)$, $\mathbf{D}=(4^{13})$ & \num{32768} & \SI{796.7}{\second} & \SI{6700}{\second} & \num{8.4} & \SI{782.7}{\second} & \SI{1645.1}{\second} & \num{2.1} \\
Generic $n=14$, $W=(2^{12},1^2)$, $\mathbf{D}=(4^{14})$ & \num{65536} & \SI{5040.1}{\second} & $\infty$ & $\infty$ & \SI{5602.3}{\second} & NA & NA \\
\hline
DLP Edwards $n=4$, $W=(2^{3},1)$, $\mathbf{D}=(8^4)$ & \num{512} & \SI{0.1}{\second} & \SI{0.1}{\second} & \num{1} & \SI{0.1}{\second} & \SI{0.1}{\second} & \num{1} \\
DLP Edwards $n=5$, $W=(2^{4},1)$, $\mathbf{D}=(16^5)$ & \num{65536} & \SI{935.4}{\second} & \SI{6461.2}{\second} & \num{6.9} & \SI{2164.4}{\second} & \SI{6935.6}{\second} & \num{3.2} \\
\hline
    % END RECEIVE ORGTBL benchFGb2
  \end{tabular}
  \label{tab:benchmarkFGb}
}\\
\subfloat[Benchmarks with Magma]{
  \begin{tabular}{lrllm{1.3cm}llm{1.3cm}}
    \hline
    System & $\deg(I)$ & $t_{\F4}$ (qh) & $t_{\F4}$ (std) & Speed-up for \F4 & $t_{\mathrm{FGLM}}$ (qh) & $t_{\mathrm{FGLM}}$ (std) & Speed-up for FGLM \\
    % BEGIN RECEIVE ORGTBL benchmagma
\hline
Generic $n=7$, $W=(1^{4},2^{3})$, $\mathbf{D}=(4^{7})$ & \num{2048} & \SI{7.9}{\second} & \SI{14}{\second} & \num{1.7} & \SI{214.2}{\second} & \SI{222.7}{\second} & \num{1} \\
Generic $n=8$, $W=(1^{4},2^{4})$, $\mathbf{D}=(4^{8})$ & \num{4096} & \SI{62.6}{\second} & \SI{138.3}{\second} & \num{2.2} & \SI{1774.7}{\second} & \SI{1797.1}{\second} & \num{1} \\
Generic $n=9$, $W=(1^{5},2^{4})$, $\mathbf{D}=(4^{9})$ & \num{16384} & \SI{3775.5}{\second} & \SI{8830.5}{\second} & \num{2.3} & $\infty$ & $\infty$ & NA \\
\hline
Generic $n=7$, $W=(2^{5},1^2)$, $\mathbf{D}=(4^{7})$ & \num{512} & \SI{0.2}{\second} & \SI{0.7}{\second} & \num{3.5} & \SI{45.5}{\second} & \SI{45.6}{\second} & \num{1} \\
Generic $n=8$, $W=(2^{6},1^2)$, $\mathbf{D}=(4^{8})$ & \num{1024} & \SI{1}{\second} & \SI{6.2}{\second} & \num{6.2} & \SI{512.3}{\second} & \SI{515.6}{\second} & \num{1} \\
Generic $n=9$, $W=(2^{7},1^2)$, $\mathbf{D}=(4^{9})$ & \num{2048} & \SI{6}{\second} & \SI{88.1}{\second} & \num{14.7} & \SI{7965}{\second} & \SI{8069.4}{\second} & \num{1} \\
Generic $n=10$, $W=(2^{8},1^2)$, $\mathbf{D}=(4^{10})$ & \num{4096} & \SI{42.4}{\second} & \SI{911.8}{\second} & \num{21.5} & $\infty$ & $\infty$ & NA \\
Generic $n=11$, $W=(2^{9},1^2)$, $\mathbf{D}=(4^{11})$ & \num{8192} & \SI{292.5}{\second} & \SI{12126.4}{\second} & \num{41.5} & $\infty$ & $\infty$ & NA \\
Generic $n=12$, $W=(2^{10},1^2)$, $\mathbf{D}=(4^{12})$ & \num{16384} & \SI{2463.2}{\second} & \SI{146774.7}{\second} & \num{59.6} & $\infty$ & $\infty$ & NA \\
Generic $n=13$, $W=(2^{11},1^2)$, $\mathbf{D}=(4^{13})$ & \num{32768} & $\infty$ & $\infty$ & NA & $\infty$ & $\infty$ & NA \\
\hline
DLP Edwards $n=4$, $W=(2^{3},1)$, $\mathbf{D}=(8^4)$ & \num{512} & \SI{1}{\second} & \SI{1}{\second} & \num{1} & \SI{1}{\second} & \SI{27}{\second} & \num{27} \\
DLP Edwards $n=5$, $W=(2^{4},1)$, $\mathbf{D}=(16^5)$ & \num{65536} & \SI{6044}{\second} & \SI{56105}{\second} & \num{9.3} & $\infty$ & $\infty$ & NA \\
\hline
    % END RECEIVE ORGTBL benchmagma
  \end{tabular}
  \label{tab:benchmarkMagma}
}

  \caption{Benchmarks with FGb and Magma for some affine systems}
  \label{tab:benchs}
}
\end{table*}

\vspace{-0.3cm}
\section{The affine case}
\label{sec:cas-affine}%\myvspace{-0.3cm}

We will now consider the case of input that do not necessarily consist of quasi-homogeneous polynomials.
One of the methods to find a \grevlex Gröbner basis of such a system is to apply \F5,
considering at $W$-degree $d$ the set of monomials having $W$-degree
lower than or equal to $d$.
This is equivalent to homogenizing the system, i.e. to adding a variable $X_{1} > \dots > X_{n} > H$, and applying the classical \F5 algorithm to this homogeneous system.
The reverse transformation is done by evaluating each polynomial at $H=1$.

However, this process makes it harder to compute the complexity of the \F5 algorithm.
The main reason is that dehomogenizing does not necessarily preserve $W$-degree, and as a consequence, it is no longer
true that running the Matrix-\F5 algorithm up to $W$-degree $d$ provides us with a basis, truncated at $W$-degree $d$.
What remains true though is that past some $W$-degree, we may obtain a Gröbner basis for the entire ideal.

Generally, we want to avoid \emph{degree falls} in the run of \F5, that is, reductions where the $W$-degree of the reductee is less than the $W$-degrees of the polynomials forming the $S$-pair.
This phenomenon is similar to reductions to zero in the quasi-homogeneous case.
It can be ruled out by considering only systems which are \emph{regular in the affine sense} (as found in \cite{Bar04} for gradings in total degree).%, that is, systems for which no such degree fall may happen.\myvspace{-0.15cm}
\begin{definition}
  Let $W$ be a system of weights, and $(f_{1},\dots,f_{n})$ be a system of not-necessarily $W$-homogeneous polynomials.
  We denote by $h_{i}$ the quasi-homogeneous component of highest $W$-degree in $f_{i}$, for any $1 \leq i \leq n$.
  We say that the sequence $(f_{i})$ is \emph{regular in the affine sense} when the sequence $(h_{i})$ is regular (in the quasi-homogeneous sense).
  We define the \emph{degree of regularity} of the ideal $\langle f_{i} \rangle$ as the degree of regularity of the ideal $\langle h_{i} \rangle$.
\end{definition}\myvspace{-0.15cm}

Since a degree fall in a run of \F5 is precisely a reduction to zero in the highest $W$-degree quasi-homogeneous components of the
system, we know that the \F5 criterion rules out all degree falls in a run of \F5 on such a regular system.
In turns, it ensures that for such a system, running Matrix-\F5 up to degree $d$ returns a $d$-Gröbner basis of $F$.

%As long as we work with such systems,
Hence we can study the complexity of
\F5 by looking at a run of Matrix-\F5 on the homogenized system.
As an example, we prove the following theorem:\myvspace{-0.15cm}
\begin{theorem}
  Let $W=(w_{1},\dots,w_{n})$ be a system of weights, and let $f_{1},\dots,f_{m}$ be a generic system of polynomials of the form $f_{i} = g_{i} + \lambda_{i}$, with $g_{i}$ $W$-homogeneous of $W$-degré $d_{i}$ and $\lambda_{i} \in \K$.
  Let $D$ be the degree of the system, $\dreg$ its degree of regularity, and $\delta$ the gcd of the $d_{i}$'s.
  We can compute a \Wgrevlex Gröbner basis of this system in time%\myvspace{-0.3cm}
  \begin{equation}
    \label{eq:40}
    O\left( \frac{\dreg}{\delta^{\omega}}M_{d,W}(n)^{\omega}\right),
  \end{equation}
  or in other words, we can divide the known complexity of the \F5 process on such a system by $\delta^{\omega}$.%$\delta^{3}$ (assuming we work with $\omega = 3$).
\end{theorem}\myvspace{-0.35cm}
\begin{proof}
  The idea is that when we homogenize the system, we can choose any
  suitable weight for $H$, not necessarily 1.  More precisely, we can
  set the weight of $H$ to be $\delta$, so that the homogenized
  polynomials become $ f_{i}^{h} = g_{i} +
  \lambda_{i}H^{d_{i}/\delta}.$
  
  Thus, assuming the computations made at section~\ref{sec:degr-regul-degr} still hold, we have the same improvements on the bound on $\dreg$ and on  the size of matrices as before, and thus we have the wanted result.

  Note that even if the initial system is generic, the homogenized system is not.
  However, one can check that if the initial system was regular in the affine sense, the homogenized system is still regular.
  Indeed, it's enough to check that no reduction to zero occur in a Matrix-\F5 run, but it is clear, since such a reduction would in particular be a degree fall.
  Also, the property of being in Noether position for the $m$ first variables is clearly kept upon homogenizing.

  As such, generically, our homogenized system is regular and in Noether position, so the previous computations indeed still hold.
\end{proof}

%\vspace{-0.5cm}
\section{Experimental results}
\label{sec:experimental-results}

We have run some benchmarks\footnote{\small All the systems we used are available online on \url{http://www-polsys.lip6.fr/\~jcf/Software/benchsqhomog.html}.}, using the FGb library and the Magma algebra software.
We present these results in Tables~\ref{tab:benchmarkFGb} and~\ref{tab:benchmarkMagma}.
The examples are chosen with increasing $n$ (number of variables and polynomials), two different classes of systems of weights $W$ and systems of $W$-degrees $D$.
With these conditions, we built a generic system of polynomials $f_{i}$ in $\FF_{65521}[\mathbf{X}]$, such that all monomials appearing in $f_{i}$ have $W$-degree at most $d_{i}$.
The last examples are systems arising in the study of the Discrete Logarithm Problem, when trying to compute the decompositions of points on an elliptic curve (see \cite{FGHR12}).
In both cases, we use a shortened notation for the systems of weights and the degrees, where for example $(2^{3},1^{2})$ means $(2,2,2,1,1)$.
The magma benchmarks were run on a machine with \SI{128}{\giga\byte}~RAM and \SI{3}{\giga\hertz}~CPU, running \hbox{Magma~v.2.17-1}.
The FGb benchmarks were run on a laptop with \SI{16}{\giga\byte}~RAM and \SI{3}{\giga\hertz}~CPU.
%The coefficient field is $\FF_{65521}$ too.

For each system, we compared our strategy (``qh'') with the default strategy (``std''), for both steps.
The algorithms used by the FGb library are \F5 and an implementation of FGLM taking advantage of the sparsity of the matrices (\cite{SparseFGLM}).
The algorithms used by Magma are \F4 and the classical FGLM.
The complexity of sparse-FGLM depends on the number of solutions of the system and on the shape of the input basis, while the complexity of classical FGLM depends only on the number of solutions.
This explains why we can see a speed-up on the FGLM step in FGb, even though the degree is unchanged.
%The empty fields in the benchmarks mean that we could not run the FGLM algorithm on that system, because the computation of the intermediate \grevlex basis failed.

\textbf{Acknowledgments.} This work was supported in part by the HPAC grant (ANR ANR-11-BS02-013) and by the EXACTA grant (ANR-09-BLAN-0371-01) of the French National Research Agency.
%The second author is member of Institut Universitaire de France and partially supported by it.

%\vspace{0.1cm}
\bibliographystyle{abbrv}
%\bibliography{MyBiblio}
%{\bibliography{issac53p-faugere}}
%%% \include{extra.tex}

\end{document}

%%% Local Variables: 
%%% mode: latex
%%% eval: (load "benchmarks-latex.el")
%%% TeX-master: t
%%% End: 